\documentclass[a4paper,UKenglish,cleveref, autoref, thm-restate]{lipics-v2021}

\pdfoutput=1 
\hideLIPIcs  

\title{Improved Approximation Ratios for the Shortest Common Superstring Problem with Reverse Complements} 

\titlerunning{Improved Approximation Ratios for SCS-RC} 

\author{Ryosuke Yamano}{Department of Computer Science, Graduate School of Information Science and Technology, The University of Tokyo, Japan \and Division of Medical Data Informatics, Human Genome Center, Institute of Medical Science, The University of Tokyo, Japan}{ryoyamano15@g.ecc.u-tokyo.ac.jp}{https://orcid.org/0009-0002-1683-5179}{}
\author{Tetsuo Shibuya}{Division of Medical Data Informatics, Human Genome Center, Institute of Medical Science, The University of Tokyo, Japan}
{tshibuya@hgc.jp}{https://orcid.org/0000-0003-1514-5766}{}

\authorrunning{R. Yamano and T. Shibuya} 

\Copyright{Ryosuke Yamano and Tetsuo Shibuya} 

\ccsdesc[100]{Theory of computation~Approximation algorithms analysis} 

\keywords{Shortest Common Superstring, Approximation Algorithms, DNA Sequencing} 

\category{} 

\relatedversion{} 

\funding{This work was supported by MEXT KAKENHI Grant Numbers 21H05052, 23H03345, and 23K18501.}

\nolinenumbers 

\newcommand{\op}[1]{\mathrm{#1}}

\usepackage[ruled]{algorithm2e}

\EventEditors{John Q. Open and Joan R. Access}
\EventNoEds{2}
\EventLongTitle{42nd Conference on Very Important Topics (CVIT 2016)}
\EventShortTitle{CVIT 2016}
\EventAcronym{CVIT}
\EventYear{2016}
\EventDate{December 24--27, 2016}
\EventLocation{Little Whinging, United Kingdom}
\EventLogo{}
\SeriesVolume{42}
\ArticleNo{23}

\begin{document}

\maketitle

\begin{abstract}
The Shortest Common Superstring (SCS) problem asks for the shortest string that contains each of a given set of strings as a substring. Its reverse-complement variant, the Shortest Common Superstring problem with Reverse Complements (SCS-RC), naturally arises in bioinformatics applications, where for each input string, either the string itself or its reverse complement must appear as a substring of the superstring. The well-known \textsf{MGREEDY} algorithm for the standard SCS constructs a superstring by first computing an optimal cycle cover on the overlap graph and then concatenating the strings corresponding to the cycles, while its refined variant, \textsf{TGREEDY}, further improves the approximation ratio. Although the original 4- and 3-approximation bounds of these algorithms have been successively improved for the standard SCS, no such progress has been made for the reverse-complement setting. A previous study extended \textsf{MGREEDY} to SCS-RC with a 4-approximation guarantee and briefly suggested that extending \textsf{TGREEDY} to the reverse-complement setting could achieve a 3-approximation. In this work, we strengthen these results by proving that the extensions of \textsf{MGREEDY} and \textsf{TGREEDY} to the reverse-complement setting achieve 3.75- and 2.875-approximation ratios, respectively. Our analysis extends the classical proofs for the standard SCS to handle the bidirectional overlaps introduced by reverse complements. These results provide the first formal improvement of approximation guarantees for SCS-RC, with the 2.875-approximate algorithm currently representing the best known bound for this problem.
\end{abstract}

\section{Introduction}
The Shortest Common Superstring (SCS) problem is a classical problem in string algorithms and combinatorial optimization. 
Given a set $S$ of strings, the goal is to find a shortest possible string that contains every string in $S$ as a substring. 
SCS has numerous applications in various fields \cite{Gevezes2014}, among which one of the most prominent is DNA sequencing. 
A DNA molecule consists of four nucleotides (Adenine, Thymine, Guanine, and Cytosine) and is assembled from short sequence reads. 
This process can be viewed as an instance of the SCS problem over a quaternary alphabet. 
However, unlike ordinary strings, DNA sequences are double-stranded, and sequencing technologies often produce reads whose orientation is not known, i.e., which of the two DNA-duplex strands they came from~\cite{MyersJr+2016+126+132}. 
Consequently, for each read, either the read itself or its reverse complement may occur in the original genome. 
To capture this property, the Shortest Common Superstring problem with Reverse Complements (SCS-RC) extends the classical SCS problem by requiring the superstring to contain, for each input string, either the string or its reverse complement as a substring. 
This extension provides a more realistic abstraction of genomic sequence reconstruction, yet its algorithmic properties and approximation guarantees have been studied far less extensively than those of the standard SCS problem.

For the standard SCS problem, Blum et al.\ \cite{Blum.et.al} showed that the \textsf{GREEDY} algorithm, which repeatedly merges two distinct strings with the maximum overlap, is 4-approximate.  
They also introduced two variants, \textsf{MGREEDY} and \textsf{TGREEDY}, and proved that these achieve 4- and 3-approximation ratios, respectively.  
Many subsequent works have further improved these approximation ratios \cite{CPM1996.2+2/3approx, Breslaure.1997.OverlapRotationLemma, SIAM.2+1/2approx, KAPLAN200513.greedy3.5n, SODA13.2+11/23approx, STOC22.ImprovedApproximateGuarantees}.  
Englert et al.\ \cite{englert_et_al:LIPIcs.ISAAC.2023.29} established that the \textsf{MGREEDY} algorithm achieves a guarantee of $\frac{\sqrt{67} + 2}{3} \approx 3.396$, and by combining this with the $\frac{2}{3}$-approximation algorithm for the maximum asymmetric traveling salesman problem \cite{Kaplan.MAXATSP.2/3approx, paluch_et_al:LIPIcs.STACS.2012.501}, they obtained the currently best-known approximation ratio for SCS, $\frac{\sqrt{67} + 14}{9} \approx 2.465$.

For the SCS-RC problem, Jiang et al.\ \cite{JIANG1992195} extended the work of Blum et al.\ \cite{Blum.et.al} to the reverse-complement setting.
They proved that the extension of the \textsf{MGREEDY} algorithm, which we refer to as \textsf{MGREEDY-RC} (shown in \cref{alg:MGREEDY-RC}), is a 4-approximate algorithm for SCS-RC.  
They further noted that the \textsf{TGREEDY} algorithm can be similarly extended (the \textsf{TGREEDY-RC} algorithm shown in \cref{alg:TGREEDY-RC}) to achieve a 3-approximation, relying on the fact that the \textsf{GREEDY-RC} algorithm (shown in \cref{alg:GREEDY-RC}) achieves a $\frac{1}{2}$ compression ratio, which was formally proved in \cite{FICI2016245}.

However, no further improvements for SCS-RC have been reported since then, leaving a substantial gap compared to the standard SCS.
Our analysis narrows this gap and provides new insights into greedy algorithms under the reverse-complement setting.

\subsection{Our Contributions}

We first establish \cref{thm:framework MGREEDY-RC to GREEDY-RC}, which provides a framework analogous to that used in the standard SCS problem \cite{STOC22.ImprovedApproximateGuarantees, englert_et_al:LIPIcs.ISAAC.2023.29}.

\begin{theorem} \label{thm:framework MGREEDY-RC to GREEDY-RC}
If \textsf{MGREEDY-RC} is a $(2 + \alpha)$-approximation algorithm and there exists an algorithm that achieves a compression ratio of $\delta$ for SCS-RC, then one can construct a $(2 + (1 - \delta)\alpha)$-approximation algorithm for SCS-RC.
In particular, \textsf{TGREEDY-RC} corresponds to the case $\delta = \frac{1}{2}$ within this general framework, which immediately yields a $(2 + \frac{\alpha}{2})$-approximation guarantee.
\end{theorem}

We extend the work of Kaplan and Shafrir \cite{KAPLAN200513.greedy3.5n} to the reverse-complement setting, improving the approximation guarantee of \textsf{MGREEDY-RC} from 4 to 3.75.
The presence of reverse complements prevents the direct application of their method, since it relies on the overlap rotation lemma (stated in \cref{lm:overlap rotation lemma}), which no longer holds when the rotated string is reverse complemented.
This highlights the additional structural challenges of the SCS-RC problem compared with the standard SCS problem.
Nevertheless, we manage to adapt part of their improvement by exploiting the properties of reverse complements.

\begin{theorem} \label{thm:MGREEDY-RC 3.75-approx}
\textsf{MGREEDY-RC} is a $3.75$-approximate algorithm.
\end{theorem}

Combining \cref{thm:framework MGREEDY-RC to GREEDY-RC} and \cref{thm:MGREEDY-RC 3.75-approx}, we immediately obtain improved guarantees for \textsf{TGREEDY-RC}, tightening its approximation ratio from $3$ to $2.875$, which currently represents the best-known bound for SCS-RC.

\begin{theorem} \label{thm:TGREEDY-RC 2.875 approx}
\textsf{TGREEDY-RC} is a $2.875$-approximate algorithm.
\end{theorem}

\section{Preliminaries}
Over the alphabet $\Sigma = \{\texttt{a}, \texttt{t}, \texttt{g}, \texttt{c}\}$, we define the \emph{complement mapping} $\bar{\cdot} : \Sigma \to \Sigma$ by
\[
\bar{\texttt{a}} = \texttt t, \quad \bar{\texttt t} = \texttt a, \quad \bar{\texttt g} = \texttt c, \quad \bar{\texttt c} = \texttt g.
\]
For a string $s = b_1 b_2 \dots b_n \in \Sigma^*$, its \emph{reverse complement} is defined as
\[
\bar{s}^R = \bar{b_n}\, \bar{b_{n-1}}\, \dots\, \bar{b_1}.
\]
In other words, $\bar{s}^R$ is obtained by reversing $s$ 
and then replacing each character $b_i \in \Sigma$ with its complement $\bar{b_i}$. 
Let $S = \{s_1, \dots, s_m\}$ be a set of strings over $\Sigma$, and define $\bar{S}^R = \{\bar{s_1}^R, \dots, \bar{s_m}^R\}$. 
Without loss of generality, we assume that $S \cup \bar{S}^R$ is \emph{substring-free}, i.e., no string in $S \cup \bar{S}^R$ is a substring of another.
For a string $x$, we use
\[
x' \in \{ x, \bar{x}^R \}
\] 
to denote either $x$ or its reverse complement.
We will use this notation consistently throughout the paper.
Let us formally define the SCS-RC problem:

\begin{definition}[Shortest Common Superstring with Reverse Complements (SCS-RC)]
\leavevmode\\
\textbf{Input:} A set of strings $S = \{s_1, \dots, s_m\}$ over an alphabet $\Sigma$.\\
\textbf{Output:} The shortest string $s$ such that for each $s_i \in S$,  
either $s_i$ or its reverse complement $\bar{s_i}^R$ appears as a substring of $s$.
\end{definition}

We call a (not necessarily shortest) string that contains, for each $s_i \in S$, either $s_i$ or its reverse complement $\bar{s_i}^R$, an \emph{approximate solution} for the SCS-RC instance $S$.

There are two common ways to measure the quality of approximation algorithms.

\begin{definition}[Length ratio and compression ratio]
Let $\op{ALG}(S)$ denote the length of the approximate solution for the SCS-RC instance $S$ produced by an algorithm $\op{ALG}$, and let $\op{OPT}(S)$ denote the length of an optimal solution.  
We define the \emph{length ratio} of the algorithm as
$\frac{\op{ALG}(S)}{\op{OPT}(S)}$,
and the \emph{compression ratio} as
$\frac{||S|| - \op{ALG}(S)}{||S|| - \op{OPT}(S)}$,
where $||S|| = \sum_{s_i \in S} |s_i|$.
\end{definition}

Since $\op{ALG}(S) \ge \op{OPT}(S)$, 
we distinguish between two notions of approximation. 
When $\delta > 1$, a $\delta$-approximation refers to the \emph{length ratio}, 
and when $\delta < 1$, it refers to the \emph{compression ratio}. 
Thus, if the type of ratio is omitted, it is determined by whether $\delta$ is greater or less than 1.

For two (not necessarily distinct) strings $x$ and $y$, 
let $v$ be the longest string such that 
$x = u v$ and $y = v w$ for some nonempty strings $u$ and $w$.  
We call $v$ the \emph{overlap} between $x$ and $y$, 
and denote its length by $|v| = \op{ov}(x, y)$.
The string $u$ is called the \emph{prefix} of $x$ with respect to $y$, 
and is denoted by $\op{pref}(x, y)$.  
We define the \emph{distance} from $x$ to $y$ as 
\[
\op{dist}(x, y) = |\op{pref}(x, y)| = |x| - \op{ov}(x, y).
\]
For a given sequence of strings $x_1, \dots, x_r$, 
we define
\[
\langle x_1, \dots, x_r \rangle 
= \op{pref}(x_1, x_2)
  \op{pref}(x_2, x_3)
  \dots
  \op{pref}(x_{r-1}, x_r)
  x_r.
\]
This string is the shortest one in which 
$x_1, \dots, x_r$ appear in this order as substrings.
As observed in \cite{Blum.et.al}, the optimal solution of SCS-RC must have the form  
$\langle s_{i_1}', \dots, s_{i_m}' \rangle$
for some permutation $i_{1}, \dots, i_{m}$ of $\{1, \dots, m\}$.

The \emph{distance graph} $G_{S} = (V, E, w)$ is the weighted complete directed graph constructed from the strings in $S \cup \bar{S}^R$.  
The vertex set is $V = S \cup \bar{S}^R$, and the edge set is $E = \{(x, y) \mid x, y \in V\}$.  
The weight of each edge $(x, y) \in E$ is defined as $w(x, y) = \op{dist}(x, y)$.  
The \emph{overlap graph} is defined analogously, except that the weight of each edge $(x, y)$ is given by $w(x, y) = \op{ov}(x, y)$.
Each path $x_1, \dots, x_r$ in $G_{S}$ corresponds to the string $\langle x_1, \dots, x_r\rangle$.  
Let $C = x_1, \dots, x_r, x_1$ be a cycle in $G_S$.  
We define the \emph{weight} of the cycle $C$ as
$w(C) = \sum_{i=1}^r \op{dist}(x_i, x_{i+1})$,
where we set $x_{r+1} = x_1$.
A \emph{cycle cover} of $G_{S}$ is a set of vertex-disjoint cycles such that, for each $1 \le i \le m$,  
exactly one of $s_i$ and $\bar{s_i}^R$ is contained in the cycles.  
A cycle cover of $G_S$ is said to be \emph{optimal} if it has the minimum total weight among all possible cycle covers.  
We denote an optimal cycle cover of $G_S$ by $\op{CYC}(G_S)$, and its total weight by $w(\op{CYC}(G_S))$.
Since an optimal solution to SCS-RC corresponds to a single cycle that contains exactly one of $s_i$ and $\bar{s_i}^R$ for each $1 \le i \le m$, 
it can be viewed as a special case of a cycle cover.  
Therefore, we have the following inequation
\begin{equation} \label{eq: w <= n}
w(\op{CYC}(G_S)) \le \op{OPT}(S).
\end{equation}

\subsection{Greedy algorithms}
\begin{algorithm}[t]
\DontPrintSemicolon
\caption{\textsf{MGREEDY-RC}}
\label{alg:MGREEDY-RC}
\SetKwInOut{Input}{Input}\SetKwInOut{Output}{Output}
\Input{A set of strings $S$}
\Output{An approximate solution for the SCS-RC instance $S$}
$T \leftarrow \emptyset$\;
\While{$S \neq \emptyset$}{
$P \leftarrow \{(x, y) \in (S \cup \bar{S}^R) \times (S \cup \bar{S}^R) \mid (x = y)~\text{or}~(x \neq y~\text{and} ~\bar{x}^R \neq y)\}$ \;
take $(x, y)$ from $P$ that maximizes the value $\op{ov}(x, y)$ \;
\uIf{x = y}{
    $S \leftarrow S \setminus \{x, \bar{x}^R\}$\;
    $T \leftarrow T \cup \{x\}$\;
}
\Else {
    $S \leftarrow S \setminus \{x, \bar{x}^R, y, \bar{y}^R\}$\;
    $S \leftarrow S \cup \{\langle x,y \rangle\}$\;
}
}
Concatenate all the strings in $T$\;
\end{algorithm}

\begin{algorithm}[t]
\DontPrintSemicolon
\caption{\textsf{GREEDY-RC}}
\SetKwInOut{Input}{Input}\SetKwInOut{Output}{Output}
\Input{A set of strings $S$}
\Output{An approximate solution for the SCS-RC instance $S$}
\While{$|S| > 1$}{
$P \leftarrow \{(x, y) \in (S \cup \bar{S}^R) \times (S \cup \bar{S}^R) \mid x \neq y~\text{and} ~\bar{x}^R \neq y)\}$ \;
take $(x, y)$ from $P$ that maximizes the value $\op{ov}(x, y)$ \;
$S \leftarrow S \setminus \{x, \bar{x}^R, y, \bar{y}^R\}$\;
$S \leftarrow S \cup \{\langle x,y \rangle\}$\;
}
Return the only element of $S$\;
\label{alg:GREEDY-RC}
\end{algorithm}

\begin{algorithm}[t]
\DontPrintSemicolon
\caption{\textsf{TGREEDY-RC}}
\SetKwInOut{Input}{Input}\SetKwInOut{Output}{Output}
\Input{A set of strings $S$}
\Output{An approximate solution for the SCS-RC instance $S$}
1. Compute the set $T$ using the \textsf{MGREEDY-RC} algorithm\;
2. Apply the \textsf{GREEDY-RC} algorithm to $T$\;
\label{alg:TGREEDY-RC}
\end{algorithm}

Jiang et al.\ \cite{JIANG1992195} introduced a greedy algorithm shown in \cref{alg:MGREEDY-RC}, which we refer to as \textsf{MGREEDY-RC}.  
Strictly speaking, they considered \emph{reversals} rather than \emph{reverse complements}, but their formulation and analysis can be adapted directly to our setting by interpreting their $s^R$ as $\bar{s}^R$.
\textsf{MGREEDY-RC} disallows merging string pairs of the form $(x, \bar{x}^R)$ for $x \ne \bar{x}^R$, 
since both (distinct) $x$ and $\bar{x}^R$ need not be included simultaneously. 
As discussed in \cite{JIANG1992195}, the \textsf{MGREEDY-RC} algorithm can be viewed as constructing a cycle cover.  
When it merges two strings $x = \langle s_{i_1}', \dots, s_{i_h}' \rangle$ and 
$y = \langle s_{j_1}', \dots, s_{j_g}' \rangle$, this corresponds to selecting an edge between $s_{i_h}'$ and $s_{j_1}'$ to create a new path $s_{i_1}', \dots, s_{i_h}', s_{j_1}', \dots, s_{j_g}'$.  
When it merges $x$ and $\bar{y}^R$, it first replaces the second path $s_{j_1}', \dots, s_{j_g}'$ with $\bar{s_{j_g}'}^R, \dots, \bar{s_{j_1}'}^R$, and then merges the paths to form $s_{i_1}', \dots, s_{i_h}', \bar{s_{j_g}'}^R, \dots, \bar{s_{j_1}'}^R$.  
When \textsf{MGREEDY-RC} moves a string $x = \langle s_{i_1}', \dots, s_{i_h}' \rangle$ from $S$ to $T$, it closes the corresponding path into a cycle $s_{i_1}', \dots, s_{i_h}', s_{i_1}'$.  
These operations together form a cycle cover of $G_S$, and it has been shown that this cover is optimal.

\begin{lemma}[\cite{JIANG1992195}]
\textsf{MGREEDY-RC} creates an optimal cycle cover $CYC(G_S)$.
\end{lemma}

Fici et al.\ \cite{FICI2016245} proved that the \textsf{GREEDY-RC} algorithm, given in \cref{alg:GREEDY-RC}, achieves a compression ratio of $\frac{1}{2}$, and that this bound is tight for the algorithm. As in \cite{JIANG1992195}, they actually considered \emph{reversals} rather than \emph{reverse complements}, but the same analysis applies when interpreting $s^R$ as $\bar{s}^R$.

\begin{lemma}[\cite{FICI2016245}] \label{lm:GREEDY-RC 0.5 compression}
\textsf{GREEDY-RC} achieves a compression ratio of $\frac{1}{2}$.
\end{lemma}

The \textsf{TGREEDY-RC} algorithm, presented in \cref{alg:TGREEDY-RC}, is a refinement of \textsf{MGREEDY-RC}.
Rather than concatenating the strings in $T$ directly, it merges them by executing \textsf{GREEDY-RC}.

\subsection{Periodicity and the overlap rotation lemma} \label{sec:prelim Periodicity}
We adopt the notation from \cite{Breslaure.1997.OverlapRotationLemma, KAPLAN200513.greedy3.5n}.  
A string $x$ is called a \emph{factor} of a string $s$ if $s = x^i y$ for some positive integer $i$ and some (possibly empty) prefix $y$ of $x$.  
We denote by $\op{factor}(s)$ the shortest such string $x$, and define the \emph{period} of $s$ as $\op{period}(s) = |\op{factor}(s)|$.  
A semi-infinite string $s$ is called \emph{periodic} if $s = x s$ for some nonempty string $x$.  
In this case, we denote the shortest such $x$ by $\op{factor}(s)$ and define $\op{period}(s) = |\op{factor}(s)|$.  

Let $x$ and $y$ be strings that are either finite or periodic semi-infinite.  
We call $x$ and $y$ \emph{equivalent} if $\op{factor}(y)$ is a cyclic shift of $\op{factor}(x)$, 
i.e., there exist strings $e$ and $f$ such that $\op{factor}(x) = ef$ and $\op{factor}(y) = fe$; otherwise, we call them \emph{inequivalent}.
The following lemmas were proved in \cite{Blum.et.al} and restated in \cite{Breslaure.1997.OverlapRotationLemma}.

\begin{lemma}[\cite{Blum.et.al}] 
Let $C = s_{i_1}', s_{i_2}', \dots, s_{i_k}', s_{i_1}'$ be a cycle in $CYC(G_S)$. Then
\begin{align*}
    \op{factor}(\langle s_{i_1}', \dots, s_{i_k}' \rangle) 
        &= \op{factor}(\langle s_{i_1}', \dots, s_{i_k}', s_{i_1}' \rangle) \\
        &= \op{pref}(s_{i_1}', s_{i_2}') \dots \op{pref}(s_{i_{k-1}}', s_{i_k}') \op{pref}(s_{i_k}', s_{i_1}'), \\
    \op{period}(\langle s_{i_1}', \dots, s_{i_k}' \rangle) &= w(C). \\
    \langle s_{i_1}', \dots, s_{i_k}', s_{i_1}' \rangle &= \op{factor}(\langle s_{i_1}', \dots, s_{i_k}' \rangle) s_{i_1}'
\end{align*}
\end{lemma}

\begin{lemma}[\cite{Blum.et.al}] \label{lm:rotation of cycle are equivalent}
The strings 
$
\langle s_{i_1}', \dots, s_{i_k}' \rangle, \ 
\langle s_{i_2}', \dots, s_{i_k}', s_{i_1}' \rangle, \ 
\dots, \ 
\langle s_{i_k}', s_{i_1}', \dots, s_{i_{k-1}}' \rangle
$
are all equivalent.
\end{lemma}

\begin{lemma}[\cite{Blum.et.al}] \label{lm:equivalent cycle mergable}
Let 
$C = s_{i_1}', s_{i_2}', \dots, s_{i_h}', s_{i_1}'$ 
and 
$D = s_{j_1}', s_{j_2}', \dots, s_{j_g}', s_{j_1}'$ 
be two distinct cycles in $G_S$. 
If $\langle s_{i_1}', s_{i_2}', \dots, s_{i_h}'\rangle$ is equivalent to $\langle  s_{j_1}', s_{j_2}', \dots, s_{j_g}'\rangle$
Then there exists a third cycle $E$ with weight $w(C)$ containing all the vertices in $C$ and $D$.
\end{lemma}

\cref{lm:equivalent cycle mergable} implies that strings taken from distinct cycles in an optimal cycle cover are inequivalent, and this property naturally extends to the reverse-complement setting.

\begin{lemma} \label{lm:cycles are inequivalent}
Let 
$C = s_{i_1}', s_{i_2}', \dots, s_{i_h}', s_{i_1}'$ 
and 
$D = s_{j_1}', s_{j_2}', \dots, s_{j_g}', s_{j_1}'$ 
be two distinct cycles in $CYC(G_S)$. 
Then the strings 
$e = \langle s_{i_1}', \dots, s_{i_h}' \rangle$ and
$f = \langle s_{j_1}', \dots, s_{j_g}' \rangle$
are inequivalent. 
Moreover, the pairs $(e, \bar{f}^R)$, $(\bar{e}^R, f)$, and $(\bar{e}^R, \bar{f}^R)$ are also inequivalent.
\end{lemma}

The following lemma was used to gain a 4-approximate bound in \cite{JIANG1992195}.

\begin{lemma}[\cite{Blum.et.al}] \label{lm:overlap weak lemma}
If strings $x$ and $y$ are inequivalent, then $\op{ov}(x, y) \le \op{period}(x) + \op{period}(y)$.
\end{lemma}

Given a semi-infinite string $\alpha = x_1x_2 \cdots$, we denote the rotation $\alpha[k] = x_kx_{k+1} \cdots$.
Breslauer et al.\ \cite{Breslaure.1997.OverlapRotationLemma} proved the following overlap rotation lemma.

\begin{lemma}[overlap rotation lemma \cite{Breslaure.1997.OverlapRotationLemma}] \label{lm:overlap rotation lemma}
Let $\alpha$ be a periodic semi-infinite string. There exists an integer $k$, such that for any finite string $s$ that is inequivalent to $\alpha$,
\[
\op{ov}(s, \alpha[k]) \le \op{period}(s) + \frac{1}{2}\op{period}(\alpha)
\]
\end{lemma}
We denote a rotation $\alpha[k]$ that satisfies \cref{lm:overlap rotation lemma} as the \emph{critical rotation}.  
The following lemma, originally proved in \cite{Breslaure.1997.OverlapRotationLemma} and restated in \cite{KAPLAN200513.greedy3.5n}, shows that it is possible to extract such a rotation from each cycle in $CYC(G_S)$.
Property (4) relies on the fact that $x_D$ and $x_C$ are inequivalent, which follows from Property (3) together with \cref{lm:rotation of cycle are equivalent,lm:cycles are inequivalent}.

\begin{lemma}[\cite{Breslaure.1997.OverlapRotationLemma}] \label{lm:extract critical rotation from cycle}
Let $C = s_{i_1}', \dots, s_{i_r}', s_{i_1}'$ be a cycle in $CYC(G_S)$. Then there exist a string $x_C$ and an index $j$ such that:
\begin{bracketenumerate}
    \item The string $\langle s_{i_{j+1}}', \dots, s_{i_r}', s_{i_1}', \dots, s_{i_j}' \rangle$ is a suffix of $x_C$.
    \item The string $x_C$ is contained in $y_C = \langle s_{i_j}', \dots, s_{i_r}', s_{i_1}', \dots, s_{i_j}' \rangle$.
    \item $x_C$ is equivalent to $\langle s_{i_{j+1}}', \dots, s_{i_r}', s_{i_1}', \dots, s_{i_j}' \rangle$.
    \item The semi-infinite string $\op{factor}(x_C)^\infty$ is the critical rotation of $\op{factor}(\langle s_{i_1}', \dots, s_{i_r}' \rangle)^\infty$.  
    Specifically, let $x_{D}$ be the string obtained by this lemma corresponding to a different cycle $D \in CYC(G_S)$; then  
    $\op{ov}(x_{D}, x_C) \le w(D) + \frac{1}{2}w(C)$.
\end{bracketenumerate}
\end{lemma}

\section{Analysis}
We begin by proving \cref{thm:framework MGREEDY-RC to GREEDY-RC},
which serves as the reverse-complement analogue of Theorem~3.1 in \cite{STOC22.ImprovedApproximateGuarantees}.
The argument follows the same ideas as the original proof.

\subsection{Proof of Theorem \ref{thm:framework MGREEDY-RC to GREEDY-RC}}
In \textsf{MGREEDY-RC}, we first compute the optimal cycle cover $CYC(G_S)$ and obtain the set of strings corresponding to its cycles, denoted by $T$ (see \cref{alg:MGREEDY-RC}).
We then observe that the optimal solution to the SCS-RC problem on this set $T$ is at most twice as long as the optimal solution to the original instance.

\begin{lemma} \label{lm:OPT(T) <= 2OPT(S)}
$\op{OPT}(T) \le \op{OPT}(S) + w(CYC(G_S)) \le 2\,\op{OPT}(S)$.
\end{lemma}

\begin{proof}
For each cycle $C = s_{i_1}', \dots, s_{i_r}', s_{i_1}'$ in $CYC(G_S)$, we select a single representative string $s_{i_1}'$.  
Let $S' \subseteq S$ denote the set of these representatives, one per cycle.  
Let $u$ be an optimal solution to the SCS-RC problem on $S'$.  
Clearly, $|u| \le \op{OPT}(S)$.  

For each string $s_{i_1}' \in S'$, either $s_{i_1}'$ or its reverse complement $\bar{s_{i_1}'}^R$ occurs in $u$.  
If $s_{i_1}'$ occurs, we replace one occurrence of it with
\[
e = \op{pref}(s_{i_1}', s_{i_2}') \dots \op{pref}(s_{i_{r-1}}', s_{i_r}') \op{pref}(s_{i_r}', s_{i_1}') s_{i_1}'.
\]
Otherwise, we replace $\bar{s_{i_1}'}^R$ with $\bar{e}^R$.  
Note that $s_{i_1}'$ is both the prefix and the suffix of $e$, and similarly $\bar{s_{i_1}'}^R$ is both the prefix and the suffix of $\bar{e}^R$.
This ensures that the replacement operation is well-defined and produces a valid superstring.

This replacement increases the length of $u$ by $w(CYC(G_S))$, yielding an approximate solution for the SCS-RC instance $T$.
\end{proof}

In practice, we do not know the optimal solution for the instance $T$. 
Instead, we can apply the given $\delta$-approximation algorithm in terms of the compression ratio. 
Let $\op{ALG}(T)$ denote the length of the solution produced by this $\delta$-approximation algorithm on the set $T$. 
By the definition of the compression ratio, we have
\[
\delta \,(||T|| - \op{OPT}(T)) \le ||T|| - \op{ALG}(T),
\]
where $||T|| = \sum_{x \in T} |x|$ denotes the total length of strings in $T$.
Note that $||T||$ is also the length of the superstring produced by \textsf{MGREEDY-RC}.  
Hence, the $(2 + \alpha)$-approximation guarantee of \textsf{MGREEDY} implies
\[
||T|| \le (2 + \alpha)\op{OPT}(S).
\]
Combining these inequalities, we obtain
\begin{align*}
 \op{ALG}(T) 
 &\le (1 - \delta)||T|| + \delta\,\op{OPT}(T) \\
 &\le (1 - \delta)(2 + \alpha)\op{OPT}(S) + 2\delta\,\op{OPT}(S) \\
 &= (2 + (1 - \delta)\alpha)\op{OPT}(S).
\end{align*}
\textsf{TGREEDY-RC} corresponds to the case where we run \textsf{GREEDY-RC} on the set $T$. 
From \cref{lm:GREEDY-RC 0.5 compression}, \textsf{GREEDY-RC} is a $\tfrac{1}{2}$-approximation algorithm in terms of compression ratio, 
that is, $\delta = \tfrac{1}{2}$.

\subsection{Proof of Theorem \ref{thm:MGREEDY-RC 3.75-approx}}
For each cycle $C = s_{i_1}', s_{i_2}', \dots, s_{i_r}', s_{i_1}'$ in $CYC(G_S)$, the existence of strings $x_C$ and $y_C$ satisfying the properties stated in \cref{lm:extract critical rotation from cycle} is guaranteed. We define
\[
A = \{x_C \mid C \in CYC(G_S)\}, \qquad ||A|| = \sum_{x \in A} |x|.
\]
By property (1) of \cref{lm:extract critical rotation from cycle}, $||A||$ provides an upper bound on $||T||$.

\begin{lemma} \label{lm:||T|| <= ||A||}
$||T|| \le ||A||$.
\end{lemma}

\begin{proof}
Consider a cycle $C = s_{i_1}', s_{i_2}', \dots, s_{i_r}', s_{i_1}'$ in $CYC(G_S)$, and let $j$ be the index whose existence is guaranteed by \cref{lm:extract critical rotation from cycle}.  
By property (1) of that lemma, we have
\[
|\langle s_{i_{j+1}}', \dots, s_{i_r}', s_{i_1}', \dots, s_{i_j}' \rangle| \le |x_C|.
\] 

Since \textsf{MGREEDY-RC} greedily merges string pairs with the maximum overlap, the cycle-closing edge from $s_{i_r}'$ to $s_{i_1}'$ has the smallest overlap among edges in the cycle. 
Hence,
\[
\op{ov}(s_{i_r}', s_{i_1}') \le \op{ov}(s_{i_k}', s_{i_{k+1}}') \quad \text{for all } 1 \le k < r.
\]

The string produced by \textsf{MGREEDY-RC} is $\langle s_{i_1}', \dots, s_{i_r}' \rangle$, whose length can be computed as
\begin{align*}
|\langle s_{i_1}', \dots, s_{i_r}' \rangle|
&= |\op{pref}(s_{i_1}', s_{i_2}') \op{pref}(s_{i_2}', s_{i_3}') \dots \op{pref}(s_{i_{r-1}}', s_{i_r}') s_{i_r}'| \\
&= |\op{pref}(s_{i_1}', s_{i_2}') \op{pref}(s_{i_2}', s_{i_3}') \dots \op{pref}(s_{i_{r-1}}', s_{i_r}') \op{pref}(s_{i_r}', s_{i_1}')| 
   + \op{ov}(s_{i_r}', s_{i_1}') \\
&= w(C) + \op{ov}(s_{i_r}', s_{i_1}').
\end{align*}

On the other hand, for the rotated sequence starting at index $j+1$, we have
\begin{align*}
|\langle s_{i_{j+1}}', \dots, s_{i_r}', s_{i_1}', \dots, s_{i_j}' \rangle|
&= |\op{pref}(s_{i_{j+1}}', s_{i_{j+2}}') \dots \op{pref}(s_{i_{r-1}}', s_{i_r}') \op{pref}(s_{i_r}', s_{i_1}')| \\
&\quad + |\op{pref}(s_{i_1}', s_{i_2}') \dots \op{pref}(s_{i_j}', s_{i_{j+1}}')| + \op{ov}(s_{i_{j}}', s_{i_{j+1}}') \\
&= w(C) + \op{ov}(s_{i_{j}}', s_{i_{j+1}}').
\end{align*}

Since $\op{ov}(s_{i_r}', s_{i_1}') \le \op{ov}(s_{i_{j}}', s_{i_{j+1}}')$, we conclude that
\[
|\langle s_{i_1}', \dots, s_{i_r}' \rangle| \le |\langle s_{i_{j+1}}', \dots, s_{i_r}', s_{i_1}', \dots, s_{i_j}' \rangle|.
\]

Summing over all cycles in $CYC(G_S)$ yields $||T|| \le ||A||$.
\end{proof}

From the previous lemma, we have established that $||T|| \le ||A||$. 
Our next goal is to obtain an upper bound on $||A||$. 
To this end, consider the SCS-RC problem on $A$, and let $\op{OV}(A)$ denote the maximal overlap in $A$, defined as
\[
\op{OV}(A) = \sum_{j=1}^{k-1} \op{ov}(x_{C_{j}}', x_{C_{j+1}}'),
\]
where $C_1, \dots, C_k$ are the $k = |CYC(G_S)|$ cycles in $CYC(G_S)$, ordered according to an optimal superstring $\langle x_{C_1}', \dots, x_{C_k}' \rangle$ for the SCS-RC instance on $A$.

Then, by definition, we have
\[
||A|| = \op{OPT}(A) + \op{OV}(A).
\]
An upper bound on $\op{OPT}(A)$ can be obtained in a manner similar to \cref{lm:OPT(T) <= 2OPT(S)}.

\begin{lemma} \label{lm:OPT(A) upperbound}
$\op{OPT}(A) \le \op{OPT}(S) + w(CYC(G_S))$.
\end{lemma}
\begin{proof}
Let $B = \{y_C \mid C \in CYC(G_S)\}$.  
By property (2) of \cref{lm:extract critical rotation from cycle}, each $y_C$ contains $x_C$, which implies $\op{OPT}(A) \le \op{OPT}(B)$.  

Let $y_C = \langle s_{i_j}', \dots, s_{i_r}', s_{i_1}', \dots, s_{i_j}' \rangle$, and select $s_{i_j}'$ as the representative of the cycle $C$.  
Let $S' \subseteq S$ denote the set of these representatives over all cycles $C \in CYC(G_S)$.  

Let $u$ be an optimal solution of SCS-RC for $S'$. Clearly, $|u| \le \op{OPT}(S)$.  
For each string $s_{i_j}' \in S'$, either $s_{i_j}'$ or its reverse complement $\bar{s_{i_j}'}^R$ occurs in $u$.  
If $s_{i_j}'$ occurs, replace once occurence of it with the original $y_C$; otherwise, replace $\bar{s_{i_j}'}^R$ with $\bar{y_C}^R$.  
This replacement increases the length of $u$ by $w(CYC(G_S))$, yielding an approximate solution for the SCS-RC instance $B$.  
Hence, $\op{OPT}(B) \le \op{OPT}(S) + w(CYC(G_S))$.  
Together with $\op{OPT}(A) \le \op{OPT}(B)$, we obtain
$\op{OPT}(A) \le \op{OPT}(S) + w(CYC(G_S))$.
\end{proof}

For each cycle $C = s_{i_1}', s_{i_2}', \dots, s_{i_r}', s_{i_1}'$ in $CYC(G_S)$,  
we can consider its reverse complement,  
$\bar{C}^R = \bar{s_{i_1}'}^R, \bar{s_{i_r}'}^R, \dots, \bar{s_{i_2}'}^R, \bar{s_{i_1}'}^R$.  
For this $\bar{C}^R$, there exists a string $x_{\bar{C}^R}$ that satisfies the properties stated in \cref{lm:extract critical rotation from cycle}.  
However, since $x_{\bar{C}^R}$ and $\bar{x_C}^R$ may differ, we cannot directly derive the upper bound $\op{OV}(A) \le 1.5\,w(CYC(G_S))$ from \cref{lm:overlap rotation lemma,lm:extract critical rotation from cycle}.

Therefore, when the reverse complement $\bar{x_C}^R$ appears in the optimal solution, we must apply \cref{lm:overlap weak lemma} instead of \cref{lm:overlap rotation lemma}.
The key observation is that, by also considering the reverse complement of the entire optimal superstring, the proportion of such strings can be limited to at most half.

\begin{lemma} \label{lm:OV(A) upperbound}
$\op{OV}(A) \le 1.75\,w(CYC(G_S))$.
\end{lemma}
\begin{proof}
Let $C_1, \dots, C_k$ be the $k = |CYC(G_S)|$ cycles in $CYC(G_S)$,
ordered according to an optimal solution for the SCS-RC instance on $A$.
Let $v = \langle x_{C_1}', \dots, x_{C_k}' \rangle$ denote this optimal superstring. Recall that for any $i \neq j$, the strings $x_{C_i}'$ and $x_{C_j}'$ are inequivalent, as discussed in \cref{sec:prelim Periodicity}.
Then,
\begin{align*}
\op{OV}(A)
&= \sum_{j=1}^{k-1} \op{ov}(x_{C_j}', x_{C_{j+1}}') \\
&\le \sum_{j=2}^{k}
\begin{cases}
w(C_{j-1}) + 0.5\,w(C_j), & \text{if } x_{C_j}' = x_{C_j} \quad \text{(by \cref{lm:overlap rotation lemma})}, \\[4pt]
w(C_{j-1}) + w(C_j), & \text{if } x_{C_j}' = \bar{x_{C_j}}^R \quad \text{(by \cref{lm:overlap weak lemma})}.
\end{cases}
\end{align*}

Rearranging the terms gives
\[
\op{OV}(A)
\le \sum_{j=1}^{k}
\begin{cases}
1.5\,w(C_j), & \text{if } x_{C_j}' = x_{C_j}, \\[4pt]
2\,w(C_j), & \text{if } x_{C_j}' = \bar{x_{C_j}}^R.
\end{cases}
\]

Next, consider the reverse complement of the optimal solution,
$\bar{v}^R = \langle \bar{x_{C_k}'}^R, \dots, \bar{x_{C_1}'}^R \rangle$,
which is also an optimal superstring for the same instance.
In this reversed solution, every occurrence of $x_{C_j}$ is replaced by
$\bar{x_{C_j}}^R$, and vice versa. Hence, we also have
\[
\op{OV}(A)
\le \sum_{j=1}^{k}
\begin{cases}
2\,w(C_j), & \text{if } x_{C_j}' = x_{C_j}, \\[4pt]
1.5\,w(C_j), & \text{if } x_{C_j}' = \bar{x_{C_j}}^R.
\end{cases}
\]
By adding the two inequalities and dividing by two, we obtain
\[
2\,\op{OV}(A) \le \sum_{j=1}^{k} 3.5\,w(C_j),
\quad\text{that is,}\quad
\op{OV}(A) \le 1.75\,w(CYC(G_S)).
\]
\end{proof}

From \cref{lm:OPT(A) upperbound,lm:OV(A) upperbound}, we obtain the following upper bound on $||A||$:
\begin{align*}
||A|| &= \op{OPT}(A) + \op{OV}(A) \\
&\le \op{OPT}(S) + w(CYC(G_S)) + 1.75\,w(CYC(G_S)) \\
&= \op{OPT}(S) + 2.75\,w(CYC(G_S)).
\end{align*}
Using \cref{lm:||T|| <= ||A||}, we then have
\[
||T|| \le ||A|| \le \op{OPT}(S) + 2.75\,w(CYC(G_S)) \le 3.75\,\op{OPT}(S),
\]
which completes the proof of \cref{thm:MGREEDY-RC 3.75-approx}.

\section{Conclusion and Future Work}
We have improved the approximation guarantees of \textsf{MGREEDY-RC} and \textsf{TGREEDY-RC} for the SCS-RC problem by leveraging the overlap rotation lemma and certain properties of reverse complements.
The remaining gaps compared to the standard SCS problem are mainly in two aspects:
\begin{bracketenumerate}
    \item The approximation guarantee of \textsf{MGREEDY-RC}.
    \item The compression ratio.
\end{bracketenumerate}

For the first aspect, we improved the approximation guarantee from $4$ to $3.75$.  
However, for the standard SCS problem, the best known ratio of the corresponding \textsf{MGREEDY} algorithm is $\frac{\sqrt{67} + 2}{3} \approx 3.396$ \cite{englert_et_al:LIPIcs.ISAAC.2023.29},  
so closing this gap remains an open research direction.

For the second aspect, we used the $\frac{1}{2}$-approximation compression ratio algorithm in \textsf{GREEDY-RC}.  
In the case of the standard SCS, the problem can be reduced to the maximum asymmetric traveling salesman problem (MAX-ATSP), which is known to have a $\frac{2}{3}$-approximation algorithm \cite{Kaplan.MAXATSP.2/3approx, paluch_et_al:LIPIcs.STACS.2012.501}.
However, for the SCS-RC problem, we must consider clusters containing two vertices each, which prevents the direct application of approximation algorithms for MAX-ATSP.  
Another research direction is to find a suitable reduction related to the MAX-ATSP problem, or to consider an entirely different approach to improve the compression ratio.



\bibliography{lipics-v2021-sample-article}

@article{Blum.et.al,
author = {Blum, Avrim and Jiang, Tao and Li, Ming and Tromp, John and Yannakakis, Mihalis},
title = {Linear approximation of shortest superstrings},
year = {1994},
issue_date = {July 1994},
publisher = {Association for Computing Machinery},
address = {New York, NY, USA},
volume = {41},
number = {4},
issn = {0004-5411},
url = {https://doi.org/10.1145/179812.179818},
doi = {10.1145/179812.179818},
journal = {J. ACM},
month = jul,
pages = {630–647},
numpages = {18}
}

@article{JIANG1992195,
title = {A note on shortest superstrings with flipping},
journal = {Information Processing Letters},
volume = {44},
number = {4},
pages = {195-199},
year = {1992},
issn = {0020-0190},
doi = {https://doi.org/10.1016/0020-0190(92)90084-9},
url = {https://www.sciencedirect.com/science/article/pii/0020019092900849},
author = {Tao Jiang and Ming Li and Ding-Zhu Du}
}

@article{FICI2016245,
title = {On the greedy algorithm for the Shortest Common Superstring problem with reversals},
journal = {Information Processing Letters},
volume = {116},
number = {3},
pages = {245-251},
year = {2016},
issn = {0020-0190},
doi = {https://doi.org/10.1016/j.ipl.2015.11.015},
url = {https://www.sciencedirect.com/science/article/pii/S0020019015002094},
author = {Gabriele Fici and Tomasz Kociumaka and Jakub Radoszewski and Wojciech Rytter and Tomasz Waleń}
}

@Inbook{Gevezes2014,
author="Gevezes, Theodoros P.
and Pitsoulis, Leonidas S.",
title="The Shortest Superstring Problem",
bookTitle="Optimization in Science and Engineering: In Honor of the 60th Birthday of Panos M. Pardalos",
year="2014",
publisher="Springer New York",
address="New York, NY",
pages="189--227",
isbn="978-1-4939-0808-0",
doi="10.1007/978-1-4939-0808-0_10",
url="https://doi.org/10.1007/978-1-4939-0808-0_10"
}

@article{MyersJr+2016+126+132,
url = {https://doi.org/10.1515/itit-2015-0047},
title = {A history of DNA sequence assembly},
author = {Eugene W. Myers Jr},
pages = {126--132},
volume = {58},
number = {3},
journal = {it - Information Technology},
doi = {doi:10.1515/itit-2015-0047},
year = {2016},
lastchecked = {2025-11-08}
}

@InProceedings{englert_et_al:LIPIcs.ISAAC.2023.29,
  author =	{Englert, Matthias and Matsakis, Nicolaos and Vesel\'{y}, Pavel},
  title =	{{Approximation Guarantees for Shortest Superstrings: Simpler and Better}},
  booktitle =	{34th International Symposium on Algorithms and Computation (ISAAC 2023)},
  pages =	{29:1--29:17},
  series =	{Leibniz International Proceedings in Informatics (LIPIcs)},
  ISBN =	{978-3-95977-289-1},
  ISSN =	{1868-8969},
  year =	{2023},
  volume =	{283},
  editor =	{Iwata, Satoru and Kakimura, Naonori},
  publisher =	{Schloss Dagstuhl -- Leibniz-Zentrum f{\"u}r Informatik},
  address =	{Dagstuhl, Germany},
  URL =		{https://drops.dagstuhl.de/entities/document/10.4230/LIPIcs.ISAAC.2023.29},
  URN =		{urn:nbn:de:0030-drops-193319},
  doi =		{10.4230/LIPIcs.ISAAC.2023.29}
}

@article{KAPLAN200513.greedy3.5n,
title = {The greedy algorithm for shortest superstrings},
journal = {Information Processing Letters},
volume = {93},
number = {1},
pages = {13-17},
year = {2005},
issn = {0020-0190},
doi = {https://doi.org/10.1016/j.ipl.2004.09.012},
url = {https://www.sciencedirect.com/science/article/pii/S0020019004002698},
author = {Haim Kaplan and Nira Shafrir}
}

@InProceedings{CPM1996.2+2/3approx,
author="Armen, Chris
and Stein, Clifford",
editor="Hirschberg, Dan
and Myers, Gene",
title="A 2 2/3-approximation algorithm for the shortest superstring problem",
booktitle="Combinatorial Pattern Matching",
year="1996",
publisher="Springer Berlin Heidelberg",
address="Berlin, Heidelberg",
pages="87--101",
isbn="978-3-540-68390-2"
}

@article{Breslaure.1997.OverlapRotationLemma,
author = {Breslauer, Dany and Jiang, Tao and Jiang, Zhigen},
title = {Rotations of Periodic Strings and Short Superstrings},
year = {1997},
issue_date = {Aug. 1997},
publisher = {Academic Press, Inc.},
address = {USA},
volume = {24},
number = {2},
issn = {0196-6774},
url = {https://doi.org/10.1006/jagm.1997.0861},
doi = {10.1006/jagm.1997.0861},
journal = {J. Algorithms},
month = aug,
pages = {340–353},
numpages = {14}
}

@article{SIAM.2+1/2approx,
author = {Sweedyk, Z.},
title = {A 2$\frac{1}{2}$-Approximation Algorithm for Shortest Superstring},
journal = {SIAM Journal on Computing},
volume = {29},
number = {3},
pages = {954-986},
year = {2000},
doi = {10.1137/S0097539796324661},
URL = {https://doi.org/10.1137/S0097539796324661},
}

@inbook{SODA13.2+11/23approx,
author = {Marcin Mucha},
title = {Lyndon Words and Short Superstrings},
year = {2013},
publisher = {Proceedings of the 2013 Annual ACM-SIAM Symposium on Discrete Algorithms (SODA)},
pages = {958-972},
doi = {10.1137/1.9781611973105.69},
URL = {https://epubs.siam.org/doi/abs/10.1137/1.9781611973105.69},
}

@inproceedings{STOC22.ImprovedApproximateGuarantees,
author = {Englert, Matthias and Matsakis, Nicolaos and Vesel\'{y}, Pavel},
title = {Improved approximation guarantees for shortest superstrings using cycle classification by overlap to length ratios},
year = {2022},
isbn = {9781450392648},
publisher = {Association for Computing Machinery},
address = {New York, NY, USA},
url = {https://doi.org/10.1145/3519935.3520001},
doi = {10.1145/3519935.3520001},
booktitle = {Proceedings of the 54th Annual ACM SIGACT Symposium on Theory of Computing},
pages = {317–330},
numpages = {14},
location = {Rome, Italy},
series = {STOC 2022}
}

@article{Kaplan.MAXATSP.2/3approx,
author = {Kaplan, Haim and Lewenstein, Moshe and Shafrir, Nira and Sviridenko, Maxim},
title = {Approximation algorithms for asymmetric TSP by decomposing directed regular multigraphs},
year = {2005},
issue_date = {July 2005},
publisher = {Association for Computing Machinery},
address = {New York, NY, USA},
volume = {52},
number = {4},
issn = {0004-5411},
url = {https://doi.org/10.1145/1082036.1082041},
doi = {10.1145/1082036.1082041},
journal = {J. ACM},
month = jul,
pages = {602–626},
numpages = {25}
}

@InProceedings{paluch_et_al:LIPIcs.STACS.2012.501,
  author =	{Paluch, Katarzyna and Elbassioni, Khaled and van Zuylen, Anke},
  title =	{{Simpler Approximation of the Maximum Asymmetric Traveling Salesman Problem}},
  booktitle =	{29th International Symposium on Theoretical Aspects of Computer Science (STACS 2012)},
  pages =	{501--506},
  series =	{Leibniz International Proceedings in Informatics (LIPIcs)},
  ISBN =	{978-3-939897-35-4},
  ISSN =	{1868-8969},
  year =	{2012},
  volume =	{14},
  editor =	{D\"{u}rr, Christoph and Wilke, Thomas},
  publisher =	{Schloss Dagstuhl -- Leibniz-Zentrum f{\"u}r Informatik},
  address =	{Dagstuhl, Germany},
  URL =		{https://drops.dagstuhl.de/entities/document/10.4230/LIPIcs.STACS.2012.501},
  URN =		{urn:nbn:de:0030-drops-34129},
  doi =		{10.4230/LIPIcs.STACS.2012.501}
}


\end{document}